%% file: localsearchMSD.tex
\documentclass[11pt,letterpaper]{article}
 
\usepackage[utf8]{inputenc} 
\usepackage{utf8math}

\usepackage{etoolbox}

\usepackage[margin=1in]{geometry}

\usepackage{amsmath,amssymb,amsbsy,bm,latexsym,mathrsfs}
\usepackage[inline]{enumitem}
\usepackage{mathtools}
\usepackage[english]{babel}

\usepackage{xspace}

\usepackage{amsthm}

\usepackage{graphicx,epsfig,color}
\usepackage{tikz}
\usepackage[margin=1cm,font=small]{caption}

\usepackage{anysize}

\usepackage{hyperref}
\hypersetup{colorlinks, linkcolor=blue!70!black, citecolor=green!50!black, urlcolor=blue!70!black}
\usepackage{xfrac}
\usepackage[vlined,ruled]{algorithm2e}

\makeatletter
\newcommand{\labeltarget}[1]{\Hy@raisedlink{\hypertarget{#1}{}}}
\makeatother

\usepackage[textsize=footnotesize]{todonotes}

\renewcommand{\epsilon}{\varepsilon}

\providecommand{\OPT}{\mathrm{OPT}}
\providecommand{\MSD}{\ensuremath{\mathrm{MSD}}\xspace}
\providecommand{\MSDf}{\ensuremath{\mathrm{MSD}+f}\xspace}

\makeatletter
\newtheorem*{rep@theorem}{\rep@title}
\newcommand{\newreptheorem}[2]{%
\newenvironment{rep#1}[1]{%
 \def\rep@title{#2 \ref{##1}}%
 \begin{rep@theorem}}%
 {\end{rep@theorem}}}
\makeatother

\theoremstyle{plain}
\newtheorem{theorem}{Theorem}
\newtheorem{lemma}[theorem]{Lemma}
\newtheorem{corollary}[theorem]{Corollary}

\newreptheorem{theorem}{Theorem}

\theoremstyle{definition}

\title{Local Search for Max-Sum Diversification}
\author{
  Alfonso Cevallos\footnote{École Polytechnique Fédérale de Lausanne (EPFL), Switzerland. \texttt{alfonso.cevallosmanzano@epfl.ch}} \quad {\small }  \quad  Friedrich Eisenbrand\footnote{\'Ecole Polytechnique Fédérale de Lausanne (EPFL), Switzerland. \texttt{friedrich.eisenbrand@epfl.ch}} \quad {\small } \quad Rico Zenklusen\footnote{Swiss Federal Institute of Technology in Zurich (ETH Zurich), Switzerland. \texttt{ricoz@math.ethz.ch}} \\[.1cm]
}\date{\today }

\begin{document}
\maketitle
\thispagestyle{empty}

\begin{abstract}

We provide simple and fast polynomial time approximation schemes (PTASs) for several variants of the \emph{max-sum diversification} problem which, in its most basic form, is as follows: Given $n$ points $p_1,\dots,p_n ∈ \mathbb{R}^d$ and an integer $k$, select $k$ points such that the average Euclidean distance between these points is maximized. This problem commonly appears in information retrieval and web-search in order to select a \emph{diverse} set of points from the input. In this context, it has recently received a lot of attention. 
  
We present new techniques to analyze natural local search algorithms. This leads to a $(1-O(\frac{1}{k}))$-approximation for distances of negative type, even subject to any matroid constraint of rank $k$, in time $O(n k^2 \log k)$, when assuming that distance evaluations and calls to the independence oracle are constant time. Negative type distances include as special cases Euclidean distances and many further natural distances. Our result easily transforms into a PTAS. It improves on the only previously known PTAS for this setting, which relies on convex optimization techniques in an $n$-dimensional space and is impractical for large data sets. In contrast, our procedure has an (optimal) linear dependence on $n$. 

Using generalized exchange properties of matroid intersection, we show that a PTAS can be obtained for matroid intersection constraints as well. Moreover, our techniques, being based on local search, are conceptually simple and allow for various extensions. In particular, we get asymptotically optimal $O(1)$-approximations when combining the classic dispersion function with a monotone submodular objective, which is a very common class of functions to measure diversity and relevance. This result leverages recent advances on local search techniques based on proxy functions to obtain optimal approximations for monotone submodular function maximization subject to a matroid constraint.

\end{abstract}

\section{Introduction}
 
When dealing with large data sets, it is often crucial to be able to extract a smaller well-diversified subset of the data.  This is a classic problem in information retrieval, and appears in many natural settings.
For example, a news website often presents the reader with a small list of highlighted stories that should be as relevant as possible to the reader. However, at the same time, the shown news stories should exhibit a certain diversity.
The analogous problem appears in the context of search engines, showing a small set of hits that are at the same time relevant and diverse.
Similarly, to get a better overview of a large dataset, it is often of
interest to exhibit a small set of entries that reflect the typical
entries found in the data set. Again, a relevant and diverse subsample
is desired.  Not surprisingly, diversity maximization became an
important concept in information retrieval, computational geometry and
operations research.

In this paper we focus on one classic diversity measure, namely the \emph{dispersion}, which has been studied in the operations research community~\cite{hassin1997approximation,ravi1994heuristic,birnbaum2009improved,fekete2004maximum} and is currently receiving considerable attention in the information retrieval literature~\cite{gollapudi2009axiomatic,bhattacharya2011consideration,borodin2012max}.
More formally, we are given a finite ground set $X$ of size $n$ and a symmetric nonnegative function $d:X\times X \rightarrow \mathbb{R}_{\geq 0}$ between pairs of $X$ satisfying $d(a,a)=0$ for $a\in X$. Such a function $f$ is called a \emph{distance function}, and we highlight that it does not necessarily have to be metric. The dispersion of a set $A\subseteq X$ is the sum of all pair-wise distances within $A$; in short, we denote the dispersion of $A$ by 
\begin{equation*}
d(A) \coloneqq \sum_{\{a,b\}\subseteq A} d(a,b)\enspace.
\end{equation*}
Diversity maximization problems with respect to the dispersion are known under several names, in particular \emph{max-sum diversification} (in short \MSD ), \emph{max dispersion} or \emph{remote clique}. In its most basic form, the task is to maximize the dispersion under a cardinality constraint of rank $k$, i.e.
\begin{equation*}
\max\{d(A) \mid A\subseteq X, |A| \leq k\}\enspace.
\end{equation*}
For brevity, we denote this max-sum diversification problem under a cardinality constraint by $\MSD_k$. A natural generalization of $\MSD_k$ is obtained by replacing the cardinality constraint with the requirement that the set $A$ must be independent in a given matroid $M=(X,\mathcal{I})$ which, as is usual, is assumed to be given by an independence oracle.\footnote{We recall that a matroid $M=(X,\mathcal{I})$ consists of a finite ground set $X$ and a non-empty family $\mathcal{I}\subseteq 2^X$ of subsets called \emph{independent sets}, satisfying: \begin{enumerate*}[label=(\roman*)] \item if $A\in \mathcal{I}, B\subseteq A \Rightarrow B \in \mathcal{I}$, and \item if $A,B\in \mathcal{I}$ with $|A| > |B|$ $\Rightarrow \exists a\in A\setminus B$ such that $B\cup \{a\}\in \mathcal{I}$  \end{enumerate*}. For more information related to matroids we refer to~\cite{schrijver2003combinatorial}.}
Matroid constraints cover natural relevant cases, for instance the items in $X$ can be partitioned into several subgroups, and a certain number of elements have to be chosen in each subgroup.

A wide set of distance functions $d$ have been considered in this context.
A very common distance type is obtained by first representing the elements of the ground set $X$ by vectors in a high-dimensional space $\mathbb{R}^q$, and then selecting a norm in $\mathbb{R}^q$ and using the corresponding induced distances (see, e.g., \cite{manning2008introduction,salton1983introduction}).

Recently, many results have been developed for various combinations of constraints, distances, and objectives. In particular, constant-factor approximations have been obtained for $\MSD_k$ for metric distances $d$. More precisely, Ravi, Rosenkrantz and Tayi~\cite{ravi1994heuristic} showed that a natural greedy procedure is a $4$-approximation, which was later shown by Birnbaum and Goldman~\cite{birnbaum2009improved} to even achieve an approximation factor of $2$, asymptotically. Prior to this improved analysis, Hassin, Rubinstein and Tamir~\cite{hassin1997approximation} presented a different $2$-approximation. An approximation factor of $2$ is tight for this problem assuming that the \emph{planted clique} problem~\cite{alon2011inapproximability} is hard (see~\cite{borodin2012max}). Moreover, Fekete and Meijer~\cite{fekete2004maximum} showed that if the distance $d$ stems from the $\ell_1$-norm in a constant-dimensional space, then a PTAS can be obtained.\footnote{We recall that a polynomial-time approximation scheme (PTAS) is an algorithm that, for any fixed $\epsilon >0$, returns a $(1-\epsilon)$-approximation in polynomial time.}

For \MSD under a matroid constraint, Abbassi, Mirrokni and Thakur~\cite{abbassi2013diversity} and Borodin, Lee, and Ye~\cite{borodin2012max} recently obtained $\frac{1}{2}$-approximations if $d$ is metric, which were the first constant-factor approximations for this setting. The approximation factor of $\frac{1}{2}$ is as well tight here, because this setting captures $\MSD_k$ with metric distances.

Motivated by various applications, interest also arose in generalized objective functions. In particular Bhattacharya, Gollapudi and Munagala~\cite{bhattacharya2011consideration} considered an objective that consists of the dispersion plus an additive linear term, which allowed for also representing \emph{scores of documents}.
Even more generally, Borodin, Lee, and Ye~\cite{borodin2012max} studied the sum of the dispersion function with a monotone submodular function, and showed that even for this generalization a $\frac{1}{2}$-approximation can be achieved for this objective under a matroid constraint via local search.

Very recently, the authors showed~\cite{cevallos_2016_max-sum} that considerably stronger results can be achieved under a matroid constraint for many frequently used distances $d$ that have the property of being of \emph{negative type}, which is defined as follows.
Let $D\in \mathbb{R}_{\geq 0}^{n\times n}$ be the distance matrix corresponding to $d$, i.e., $D_{a,b}=d(a,b)$ for $a,b\in X$. Then, $d$ is of negative type if
\begin{equation*}
x^T D x \leq 0 \qquad \forall x\in \mathbb{R}^n \text{ with } \sum_{i=1}^n x_i=0\enspace.
\end{equation*}
Commonly used distances of negative type include the ones induced by $\ell_1$ and $\ell_2$ norms, the cosine distance or the Jaccard distance~\cite{Pekalska:2005:DRP:1197035}.
For more information on negative type distances, we refer the interested reader to~\cite{schoenberg1938metricI,schoenberg1938metricII,deza1990metric,DezaLaurent97}.
Norms corresponding to negative distance types have as well been used for similarity measures via lower-dimensional bit vectors stemming from sketching techniques~\cite{charikar2002similarity}.
For negative type distances, a PTAS for \MSD can be achieved under a matroid constraint, based on rounding solutions obtained by a convex relaxation of the problem~\cite{cevallos_2016_max-sum}.
Whereas this is essentially optimal in terms of approximation quality, the employed technique requires to solve $n$-dimensional convex optimization problems, which is impractical for large data sets and large data sets are usually to be dealt with in web-search and information retrieval. This motivates the guiding questions of this paper. 
\begin{enumerate}
\item Can one profit from the additional structure of negative type distances to obtain \emph{efficient} PTASs  that are suitable for \emph{large-scale} problems? 
\item Can one obtain such algorithms in other, more general relevant  settings,  beyond matroid constraints? 
\end{enumerate}

%
%
%

%
%
%

\subsection*{Our results}

Our key contribution is the analysis of conceptually simple and
classic local search techniques for \MSD with negative type distances.
These procedures are easy to implement and considerably faster than the convex optimization approach suggested in~\cite{cevallos_2016_max-sum}.
First, this gives us a strong approximation algorithm for matroid constraints, which implies a PTAS.  

\begin{theorem}\label{thm:PTASmat}
There is a $(1-\frac{5}{k})$-approximation for \MSD with negative type distances subject to a matroid constraint of rank $k$, running in $O(n k^2 \log k)$ time, when counting distance evaluations and calls to the independence oracle as unit time.
\end{theorem}
The above theorem indeed implies a PTAS: This is clear if $k\geq \frac{5}{\epsilon}$, where $\epsilon>0$ is the error parameter; otherwise, feasible solutions only have constant size and one can enumerate over all possible solutions. This running time is \emph{linear} in $n$. In light of the fact that $k$ is much smaller than $n$  this running time is very desirable for large $n$.

\medskip

Furthermore, our techniques  show  that local search yields   a PTAS  for \MSD with negative type distances even for matroid intersection constraints. Whether a PTAS exists in this setting was not known before. 

\begin{theorem}\label{thm:PTASmatInt}
There is a PTAS for \MSD with negative type distances subject to a matroid intersection constraint.
\end{theorem}
%
%
Our PTAS is a variation of a classic local search technique, which was used by Lee, Sviridenko, and Vondr\'ak~\cite{lee_2010_submodular} in the context of submodular maximization, which considers at each iteration exchanges of constant size. This is the first non-trivial approximation result for matroid intersection constraints in this context, and its purpose is mainly to show that very strong approximation guarantees are possible for this constraint type.
However, due to these larger-size exchanges, the suggested algorithm may have a large (polynomial) running time, depending on the size of the problem and the error parameter.

\medskip

Finally, 
by  combining our results with recent non-oblivious local search techniques by Filmus and Ward~\cite{filmus2014monotone}, and Sviridenko, Vondr{\'a}k and Ward~\cite{sviridenko2015optimal}, we obtain close-to-optimal approximation guarantees for \MSD with an objective function being a sum of the dispersion for negative type distances and a monotone submodular function, subject to a matroid constraint. Such objectives allow for balancing diversity and \emph{relevance}. Combinations of diversity and a linear function  have  been studied previously, see~\cite{bhattacharya2011consideration,borodin2012max}.  
The approximation factor we obtain depends on the \emph{curvature} of the submodular function, which yields a smooth interpolation between linear functions and submodular functions. We give a formal definition of curvature in Section~\ref{sec:submodObj}. 

\begin{theorem}\label{thm:msd+f}
Consider \MSD subject to a matroid constraint of rank $k$, with respect to an objective $g=d+f$, where $d$ is the dispersion for a negative type distance and $f$ is a nonnegative, monotone submodular function of curvature $c$.
Let $\lambda_d=\frac{d(\OPT)}{g(\OPT)}$ and $\lambda_f=\frac{f(\OPT)}{g(\OPT)}$, where $\OPT$ is an optimal solution to the problem of maximizing $g$ subject to the matroid constraint. Then, for any $\epsilon >0$, there is an efficient algorithm returning a solution of approximation guarantee
\begin{equation*}
1-\lambda_d \frac{4}{k}-\lambda_f \frac{c}{e} - \epsilon \geq 1-\max\left\{\frac{4}{k}, \frac{c}{e}\right\} - \epsilon \enspace.
\end{equation*}
\end{theorem}

The above result implies a PTAS for the case of $f$ being linear. This greatly improves upon the known  $1/2$-approximation~\cite{bhattacharya2011consideration,borodin2012max} if the distances are of negative type. If  $k$ is large enough, the result yields an approximation factor of $1-\frac{c}{e}-\epsilon$, which is known to be optimal even for the special case of maximizing a monotone submodular function with curvature $c$ over a matroid constraint~\cite{sviridenko2015optimal}.

\subsection*{Further related results and implications}

A further approach to deal with diversity maximization in large data sets is the computation of so-called \emph{core-sets}. A core-set is a small subset of the data such that an optimal solution to certain optimization problems, when applied only on the elements of the core-set, is close to the global optimal solution. Core-sets and the generalized notion of \emph{composable core-sets} have recently received considerable attention in the context of diversity maximization, and allow for transforming sequential algorithms into algorithms that run in MapReduce and Streaming models. We refer the interested reader to~\cite{indyk2014composable} and~\cite{ceccarello2016mapreduce} and references therein. We only mention the following direct consequence of the results in~\cite{ceccarello2016mapreduce} and Theorem~\ref{thm:PTASmat}. 

\begin{corollary}
Consider $\MSD_k$ on distances of negative type and doubling dimension $q$, such as Euclidean or Manhattan distances in $\mathbb{R}^q$. Then, for any $\epsilon >0$, there is a single-pass streaming algorithm that achieves an approximation guarantee of $1-\frac{5}{k}-\epsilon$, in space $O(\epsilon^{-q}k^2)$. 
\end{corollary}

A related geometric measure of dispersion in $\mathbb{R}^q$ is the \emph{volume} of the selected data set. An optimal constant factor approximation  algorithm for this measure  was recently given by Nikolov~\cite{nikolov2015randomized}.    

\subsection*{Organization of the paper}

In Section~\ref{sec:matConstraints}, we present our results for \MSD with respect to negative type distances and subject to a matroid constraint. In particular, this section highlights the main strategy we employ to exploit the property of negative type distances.
Section~\ref{sec:matIntConstraints} presents a PTAS for matroid intersection constraints. Finally, Section~\ref{sec:submodObj} contains our results for generalized objective functions, being a sum of the dispersion and a nonnegative, monotone submodular function.

\section{A local-search based PTAS for matroid constraints}
\label{sec:matConstraints}

Recall that the dispersion of a set $A\in X$ is $d(A):=\sum_{\{a,a'\}\subseteq A} d(a,a')=\frac{1}{2} \sum_{a,a'\in A} d(a,a')$. And we define the following auxiliary function, which represents the total sum of distances between two sets: For sets $A,B\in X$, let
$$d(A,B):= \sum_{a\in A, b\in B} d(a,b).$$
Notice in particular that $d(A,A)=2d(A)$ for any $A\in X$.

Throughout this section, we consider a matroid $M=(X,\mathcal{I})$ of rank $k\in \mathbb{Z}_{\geq 2}$.
The algorithm we analyze in this section, highlighted below as Algorithm~\ref{alg:localSearchMat}, is a well-known canonical local search algorithm that starts with a basis $A$ and iteratively considers exchanges of a single element, always maintaining a basis. For brevity we use the shorthand $A-a+b$ for $(A\setminus \{a\})\cup \{b\}$, where $A$ is a set and $a,b$ are two elements.

\smallskip

{
\begin{algorithm}[H]

\For{$i=1 \ldots \ell$}{
  \If{$\exists$ pair $(a,b)\in A\times (X\setminus A)$ such that
$A-a+b\in \mathcal{I}$ and $d(A-a+b)>d(A)$}{
\smallskip
Find such a pair $(a,b)$ maximizing $d(A-a+b)$.\\
Set $A=A-a+b$.
}
}

\textbf{return} $A$.
\smallskip

\caption{Local search for matroids, starting with a basis $A$ and running for $\ell$ iterations.}
\label{alg:localSearchMat}
\end{algorithm}
}

\medskip

The following lemma provides a key inequality for negative type distances that we exploit. As an intuitive special case, the inequality says that for any two sets $A$ and $B$ of same cardinality, the dispersion within $A$ plus the dispersion within $B$ is no more than the sum of all distances between $A$ and $B$.

\begin{lemma}
\label{lem:negTypeIneq}
  Let $d:X\times X \rightarrow \mathbb{R}_{\geq 0}$ be a distance of negative type and $D \in \mathbb{R}_{≥0}^{X \times X}$ be the corresponding matrix, i.e., $D_{a,b}=d(a,b)$ for $a,b\in X$. For any two non-zero vectors $x,y \in \mathbb{R}^X_{≥0}$ one has 
  \begin{equation}
    \label{eq:negTypeIneqVec}
   \frac{\|y\|_1}{\|x\|_1} x^T Dx  + \frac{\|x\|_1}{\|y\|_1}y^T Dy
      \leq 2 x^T D y\enspace,
  \end{equation}
  and consequently for two non-empty sets $A,B \subseteq X$ 
  \begin{equation}
    \label{eq:negTypeIneqSet}
    \frac{|B|}{|A|} d(A) + \frac{|A|}{|B|} d(B) \leq d(A,B)\enspace . 
  \end{equation}
\end{lemma}
\begin{proof}
Let $z=\|y\|_1 x - \|x\|_1 y$. Notice that $\sum_{a\in X}z(a)=0$. Hence, by the fact that $D$ is a distance of negative type we have
\begin{align*}
0 \geq z^T D z = \|y\|_1^2 x^T D x - 2 \|x\|_1 \|y\|_1 x^T D y + \|x\|_1^2 y^T D y.
\end{align*}
Inequality~\eqref{eq:negTypeIneqVec} now follows by dividing both sides of the above inequality by $\|x\|_1\|y\|_1$ and rearranging terms.
Inequality~\eqref{eq:negTypeIneqSet} follows from~\eqref{eq:negTypeIneqVec} by setting $x=\chi^A$ and $y=\chi^B$ to the characteristic vectors of $A$ and $B$, respectively, and dividing both the left-hand side and right-hand side by $2$.
\end{proof}

To analyze Algorithm~\ref{alg:localSearchMat}, we will consider a pairing of the elements of $A$ with the elements of an optimal solution $\OPT$ to the problem. It is well known that for any two bases $A,B$ of a matroid, there exists a pairing that can be used for exchanges in $A$. More precisely, removing any element of $A$ and adding its paired counterpart in $B$ will again lead to a basis. We denote such a pairing, which can be formalized as a bijection $\pi: A \rightarrow B$ and whose properties are stated in the lemma below, a \emph{Brualdi bijection}.

\begin{lemma}[Brualdi~\cite{brualdi1969comments}]\label{lem:brualdi}
For any two independent sets $A,B\in \mathcal{I}$ of equal cardinality, there is a bijection $\pi:A\rightarrow B$ such that for any $a\in A$, we have $A-a+\pi(b)\in \mathcal{I}$. In particular, such a bijection satisfies that it is the identity mapping on $A\cap B$.
\end{lemma}

To show that Algorithm~\ref{alg:localSearchMat} makes sufficient progress as long as $d(A)$ is much smaller than $d(\OPT)$, we consider a Brualdi bijection between $A$ and $\OPT$. As our analysis later shows, the distances between the pairs of a Brualdi bijection will be an error term that we have to bound. This is done by the following lemma.

\begin{lemma}\label{lem:cheapMatching}
For any two sets $A,B\subseteq X$ of equal cardinality $k$, and any bijection $\pi:A \rightarrow B$,
\begin{equation*}
\sum_{a\in A} d(a,\pi(a))\leq \frac{2}{k}d(A,B).
\end{equation*}
\end{lemma}
\begin{proof}
For any $a\in A$ with $a\neq \pi(a)$, Lemma~\ref{lem:negTypeIneq} implies
\begin{equation*}
d(A,a)+d(A,\pi(a))=d(A,\{a,\pi(a)\})\geq \frac{2}{k}d(A)
+\frac{k}{2}d(a,\pi(a)).
\end{equation*}
Notice that the above inequality is also true if $a=\pi(a)$, in which case the inequality reduces to $d(A,a) \geq \frac{1}{k}d(A)$, which is again a direct consequence of Lemma~\ref{lem:negTypeIneq}.
Summing these inequalities over $A$ gives
\begin{equation*}
d(A,A)+d(A,B)\geq 2d(A)+\frac{k}{2}\sum_{a\in A} d(a,\pi(a))\enspace .
\end{equation*}
The terms $d(A,A)$ and $2d(A)$ cancel out, and the claim follows.
\end{proof}

The following lemma shows that a locally optimal solution with respect to the considered exchange steps is a $(1-\frac{4}{k})$-approximation, without going into the details of how fast we will approach a local optimum.

\begin{lemma}\label{lem:matAverageImp}
Let $A$ and $B$ be bases of $M$, and let $\pi:A\rightarrow B$ be a bijection satisfying $\pi(a) = a$ for $a\in A\cap B$. Then
\begin{align}
\sum_{a\in A}\left( d(A-a+\pi(a)) - d(A)\right)
  &\geq \left(1-\frac{2}{k}\right) d(B) - \left(1+\frac{2}{k}\right)d(A)\enspace,\label{eq:matAverageImpDetailed}\\
\intertext{which, if $d(B) \geq d(A)$, implies}
\sum_{a\in A}\left( d(A-a+\pi(a)) - d(A)\right)
  &\geq \left(1-\frac{4}{k}\right) d(B) - d(A)\enspace.
\label{eq:matAverageImpOpt}
\end{align}

\end{lemma}
\begin{proof}
We first observe that
\begin{equation}\label{eq:expandUpdate}
d(A-a+b) = d(A) + d(A,b) - d(a,b) - d(a,A)
\qquad \forall a\in A, b\in X\setminus (A-a)\enspace.
\end{equation}
The above equation clearly holds if $a=b$, and for $a\neq b$ it can easily be checked by observing that the distance of any pair is counted the same number of times in the right-hand side and left-hand side of the equation (see Figure~\ref{fig:dAmapb}). 

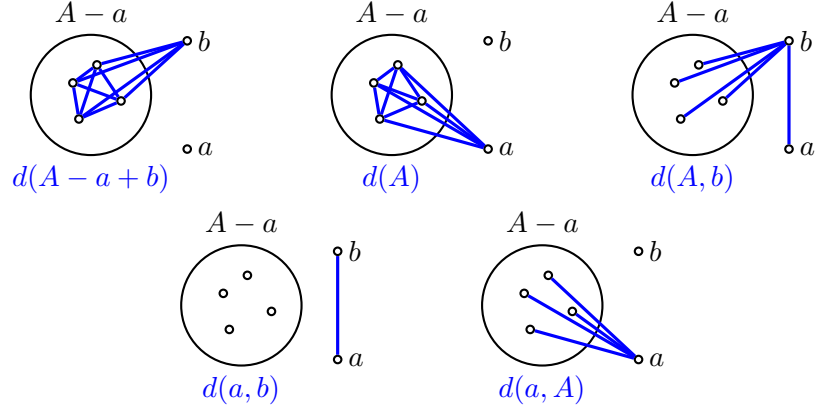
\begin{figure}[h]
\begin{center}
\input{dAmapb.tikz}
\end{center}
\caption{Graphical highlighting of the edges counted in the different terms of~\eqref{eq:expandUpdate}, for $a\neq b$. One can easily observe that each edge is counted the same number of times in the left-hand side and right-hand side of~\eqref{eq:expandUpdate}.}\label{fig:dAmapb}
\end{figure}

We finally obtain
\begin{align*}
\sum_{a\in A} \left( d(A-a+\pi(a)) - d(A)\right)
  &= \sum_{a\in A} \left( d(A,\pi(a)) - d(a,\pi(a)) - d(a,A) \right)
     && \text{(by~\eqref{eq:expandUpdate})}\\
  &= d(A,B) - d(A,A) - \sum_{a\in A}d(a,\pi(a))\\
  &\geq \left(1-\frac{2}{k}\right) d(A,B) - d(A,A) 
     && \text{(by Lemma~\ref{lem:cheapMatching})}\\
  &\geq \left(1-\frac{2}{k}\right) (d(A) + d(B)) - 2d(A)
     && \text{(by~\eqref{eq:negTypeIneqSet})}\\
  &= \left(1-\frac{2}{k}\right) d(B) - \left(1+\frac{2}{k}\right) d(A).
%
\end{align*}
Inequality~\eqref{eq:matAverageImpOpt}, valid for $d(B) \geq d(A)$, now easily follows by replacing in the right-hand side the term $-\frac{2}{k} d(A)$ by $-\frac{2}{k} d(B)$.
\end{proof}

The next theorem is a mostly standard argument showing exponentially fast convergence of the local search algorithm.

\begin{theorem}\label{thm:convLSMat}
Let $A$ be any basis of $M$.
Running Algorithm~\ref{alg:localSearchMat} for $\ell\in \mathbb{Z}_{\geq 0}$ iterations, a basis $A_\ell$ is returned such that
\begin{equation*}
d(A_\ell) \geq \left(
1 -
\left(1-\frac{1}{k}\right)^\ell \right)
\left(1-\frac{4}{k}\right)\cdot d(\OPT) \enspace .
\end{equation*}
\end{theorem}
\begin{proof}
Let $A_0 = A$ be the starting basis and we denote by $A_i$ for $i\in \{0,\ldots, \ell\}$ the basis obtained after $i$ iterations of the local search algorithm. Let $i\in \{0,\ldots, \ell-1\}$, and we consider the improvement from $A_i$ to $A_{i+1}$. Let $\pi: A_i \rightarrow \OPT$ be a Brualdi bijection. By inequality~\eqref{eq:matAverageImpOpt} in Lemma~\ref{lem:matAverageImp}, the average improvement in the dispersion of $A_i$ over the exchanges corresponding to $\pi$ is at least
\begin{equation*}
\frac{1}{k} \sum_{a\in A_i}\left(d(A_i-a+\pi(a)) - d(A_i) \right)
   \geq \frac{1}{k}\left(\left(1-\frac{4}{k}\right) d(\OPT) - d(A_i)\right)
\enspace.
\end{equation*}
Since our local search algorithm does an exchange that maximizes $d(A_{i+1}) - d(A_i)$, we have
\begin{equation*}
d(A_{i+1}) - d(A_i)
   \geq \frac{1}{k}\left(\left(1-\frac{4}{k}\right) d(\OPT) - d(A_i)\right)
      \enspace ,
\end{equation*}
which, by regrouping terms, leads to
\begin{equation*}
\left(1 - \frac{4}{k}\right)d(\OPT) - d(A_{i+1})
  \leq \left(1-\frac{1}{k}\right) \left(
    \left(1-\frac{4}{k}\right) d(\OPT) - d(A_i)
  \right)
  \quad \forall i\in \{0,\ldots, \ell-1\}
  \enspace .
\end{equation*}
The above family of inequalities implies
\begin{align*}
\left(1-\frac{4}{k}\right) d(\OPT) - d(A_{\ell})
&\leq
\left(1 - \frac{1}{k}\right)^\ell \left(
\left(1-\frac{4}{k}\right) d(\OPT) - d(A_0)
\right)\\
&\leq \left(1-\frac{1}{k}\right)^\ell \left(1-\frac{4}{k}\right) d(\OPT),
\end{align*}
which shows the theorem.
\end{proof}

Putting all ingredients together, we get our main result for \MSD with negative type distances and subject to a matroid constraint, which shows Theorem~\ref{thm:PTASmat}.

\begin{theorem}
It suffices to run Algorithm~\ref{alg:localSearchMat} for $O(k \log k)$ iterations starting with an arbitrary basis to obtain a $(1-\frac{5}{k})$-approximation for \MSD with respect to a negative type distance and subject to a matroid constraint of rank $k$.
Moreover, this algorithm can be implemented to run in $O(n k^2 \log k)$ time, when counting distance evaluations and calls to the independence oracle as unit time.
\end{theorem}
\begin{proof}
By Theorem~\ref{thm:convLSMat}, to obtain a $(1-\frac{5}{k})$-approximation, it suffices to choose the number of iterations $\ell$ of Algorithm~\ref{alg:localSearchMat} such that $(1-\frac{1}{k})^\ell \leq \frac{1}{k}$, which can be achieved by setting $\ell= O(k \log k)$.

To complete the proof, it remains to show that each iteration can be performed in $O(nk)$ time. One way to implement our local search algorithm to get the $O(nk)$ running time per iteration is as follows. At the beginning of each iteration, with current set $A$, we compute $d(A)$, and also $d(A,a)$ for each $a\in A$. For this we need $O(k^2)$ distance evaluations; we recall that $O(k^2)=O(n k)$ because $k\leq n$. We then go through the elements $b\in X\setminus A$, and for each $b\in X\setminus A$, we consider all sets $A-a+b$ for $a\in A, a \neq b$. Since $|X\setminus A| = O(n)$, it suffices to show that for a fixed $b\in X\setminus A$, we can compute the best exchange step involving $b$ in $O(k)$ time. For a fixed $b$, we first compute $d(A,b)$ in $O(k)$ time. Then, for each $a\in A,a\neq b$ we first call the independence oracle to determine whether $A-a+b\in \mathcal{I}$, and if so we compute $d(A-a+b)$ using~\eqref{eq:expandUpdate}, i.e.,
\begin{equation*}
d(A-a+b) = d(A) + d(A,b) - d(a,b) - d(a,A)\enspace .
\end{equation*}
The only quantity on the right-hand side that we did not compute so far is $d(a,b)$, which is obtained by a single distance evaluation. Hence, computing $d(A-a+b)$ for all $a\in A,a\neq b$ for which $A-a+b\in \mathcal{I}$, takes $O(k)$ time as desired, when counting distance evaluations and calls to the independence oracle as unit time.
\end{proof}

\section{A PTAS for matroid intersection}
\label{sec:matIntConstraints}

Throughout this section, $M_1=(X,\mathcal{I}_1)$ and $M_2=(X,\mathcal{I}_2)$ are two matroids on the same ground set $X$, and $k$ is the maximum cardinality of a common independent set.
The algorithm we analyze is a natural local search algorithm considering exchanges up to a certain size. It is almost identical to a procedure suggested by Lee, Sviridenko and Vondr\'{a}k~\cite{lee_2010_submodular}, designed for maximizing a monotone submodular function subject to multiple matroid constraints. The only difference is that after each exchange step, we augment the current set $A$ to a (inclusion-wise) maximal set in $\mathcal{I}_1 \cap \mathcal{I}_2$, i.e., we replace $A$ by a maximal set $A'\in \mathcal{I}_1\cap \mathcal{I}_2$ that contains $A$. It is well-known that any maximal common independent set has cardinality at least $\frac{k}{2}$. This is a property we exploit in our analysis.

For better readability, we did not try to optimize constants and often use rather loose bounds for simplicity.

\medskip

{
\begin{algorithm}[H]

\For{$i=1 \ldots \ell$}{
  \If{$\exists$ pair $(S,T)\in 2^A\times 2^{X\setminus A}$ with 
     \begin{enumerate}[label=(\roman*), nosep, leftmargin=3em]
       \item $|S|\leq p, |T|\leq p-1$,
       \item $(A\setminus S) \cup T \in \mathcal{I}_1\cap \mathcal{I}_2$, and
       \item $d((A\setminus S)\cup T)>d(A)$,
     \end{enumerate}
}
{
\smallskip
Find such a pair $(S,T)$ maximizing $d((A\setminus S)\cup T)$.\\
\smallskip
Set $A=(A\setminus S) \cup T$.\\
\smallskip
Augment $A$ to a maximal element in $\mathcal{I}_1\cap \mathcal{I}_2$.
}
}

\textbf{return} $A$.
\smallskip

\caption{Local search for matroid intersection with exchange parameter $p\in \mathbb{Z}_{\geq 2}$, starting with $A\in \mathcal{I}_1 \cap \mathcal{I}_2$ and running for $\ell$ iterations.}
\label{alg:localSearchMatInt}
\end{algorithm}
}

\medskip

Our analysis heavily relies on recently developed exchange properties for matroid intersection. The following lemma was shown in~\cite{chekuri_2011_multibudgeted} building up on previous work~\cite{lee_2010_submodular,chekuri_2010_dependent}. We state a slightly simplified version of the lemma. The original lemma provided further properties on the sets $P_i$ and also guaranteed that such sets can be found efficiently. These are properties we do not need in our analysis. The set operator $\Delta$ denotes the symmetric difference, i.e., $A\Delta B = (A\setminus B) \cup (B \setminus A)$.

\begin{lemma}[{\cite[Lemma 3.3]{chekuri_2011_multibudgeted}}]
\label{lem:matIntExchange}
Let $p\in \mathbb{Z}_{\geq 2}$, and let $A, B \in \mathcal{I}_1\cap \mathcal{I}_2$ with $|A|=|B|$. Then there exists a family of nonempty sets $P_1, \dots, P_m\subseteq X$ with $|P_i\cap A|\leq p$ and $|P_i\cap B|\leq p-1$ for $i\in [m]:=\{1,\ldots m\}$, and coefficients $\lambda_i > 0$ such that
\begin{enumerate}[label=(\roman*),itemsep=0em]
\item $A\Delta P_i \in \mathcal{I}_1\cap \mathcal{I}_2$ \;$\forall i \in [m]$, and
\item $\sum_{i=1}^m \lambda_i \chi^{P_i} = \frac{p}{p-1} \chi^{A\setminus B} + \chi^{B\setminus A}$.
\end{enumerate}
\end{lemma}

It is not hard to see that the requirement $|A|=|B|$ in Lemma~\ref{lem:matIntExchange} is not necessary. This condition was important in the original (stronger) lemma presented in~\cite{chekuri_2011_multibudgeted} which contained further properties that relied on it. For completeness, we state the lemma without the equal cardinality requirement below and provide a short proof for it.

\begin{lemma}\label{lem:matIntExchangeGen}
Let $p\in \mathbb{Z}_{\geq 2}$, and let $A, B \in \mathcal{I}_1\cap \mathcal{I}_2$. Then there exists a family of nonempty sets $P_1, \dots, P_m\subseteq X$ with $|P_i\cap A|\leq p$ and $|P_i\cap B|\leq p-1$ for $i\in [m]$, and coefficients $\lambda_i > 0$ such that
\begin{enumerate}[label=(\roman*),itemsep=0em]
\item $A\Delta P_i \in \mathcal{I}_1\cap \mathcal{I}_2$ \;$\forall i \in [m]$, and
\item\label{item:exchangeLoad} $\sum_{i=1}^m \lambda_i \chi^{P_i} = \frac{p}{p-1} \chi^{A\setminus B} + \chi^{B\setminus A}$.
\end{enumerate}
\end{lemma}

\begin{proof}
To prove the lemma, we will first ``lift'' the sets $A,B\in \mathcal{I}_1\cap \mathcal{I}_2$, which may have different sizes, to larger sets $A',B'$ of the same size that are common independent sets of two auxiliary matroids $M_1'$ and $M_2'$. We can then apply Lemma~\ref{lem:matIntExchange} to $A',B'$ with respect to these auxiliary matroids, which will imply Lemma~\ref{lem:matIntExchangeGen}.

Let $k$ be the maximum cardinality of a common independent set in $M_1$ and $M_2$, i.e., $k=\max\{|I| \mid I\in \mathcal{I}_1\cap \mathcal{I}_2\}$. Let $\bar{X}$ be a finite set of size $k$ that is disjoint from $X$. We define two auxiliary matroids $M_1'=(X',\mathcal{I}_1')$ and $M_2'=(X',\mathcal{I}_2')$ on ground set $X' = X\cup \bar{X}$. To this end, let $\bar{M}=(\bar{X},\bar{\mathcal{I}})$ be the free matroid over $\bar{X}$, i.e., $\bar{\mathcal{I}} = 2^{\bar{X}}$. For $j\in \{1,2\}$, the matroid $M'_j$ is defined to be the disjoint union of $M_j$ and $\bar{M}$. More formally, for $j\in \{1,2\}$ we have
\begin{equation*}
\mathcal{I}_j' = \{ I_j \cup \bar{I} \mid I_j \in \mathcal{I}_j, \bar{I}\in \bar{\mathcal{I}}\}.
\end{equation*}
The thus defined $M_1'$ and $M_2'$ are indeed matroids because the union of matroids again leads to a matroid (see, e.g.,~\cite[volume B]{schrijver2003combinatorial}).

Given $A,B\in \mathcal{I}_1\cap \mathcal{I}_2$, let $A'=A\cup \bar{S}_A$ and $B'=B\cup \bar{S}_B$, where $\bar{S}_A$ and $\bar{S}_B$ are arbitrary subsets of $\bar{X}$ of size $k-|A|$ and $k-|B|$, respectively. Clearly, we have $A',B'\in \mathcal{I}_1'\cap \mathcal{I}_2'$, and $|A'|=|B'|=k$. We can thus apply Lemma~\ref{lem:matIntExchange} to $A'$ and $B'$ with respect to the matroids $M_1'$ and $M_2'$ to obtain a family $P_1',\dots P_m'\subseteq X'$ with coefficients $\lambda_i > 0$ for $i\in [m]$ satisfying the properties stated in Lemma~\ref{lem:matIntExchange}. It remains to observe that the family defined by $P_i = P'_i \cap X$ for $i\in [m]$ (after removing empty sets), satisfies the properties of Lemma~\ref{lem:matIntExchangeGen}. This immediately follows from the definitions of $M_1'$ and $M_2'$ which imply that for any $S'\in \mathcal{I}'_1\cap \mathcal{I}'_2$, we have $S'\cap X \in \mathcal{I}_1\cap \mathcal{I}_2$.
\end{proof}

As in the previous section, the dispersion within the exchange sets we consider will appear as an error term in our analysis. The following lemma bounds this quantity.

\begin{lemma}\label{lem:boundWithinPi}
Let $p\in \mathbb{Z}_{\geq 2}$, and let $A,B\in \mathcal{I}_1\cap \mathcal{I}_2$. Moreover, let $P_i\subseteq X$ for $i\in [m]$ and $\lambda_i > 0$ for $i\in [m]$ be a family of sets and coefficients satisfying the conditions of Lemma~\ref{lem:matIntExchangeGen}. Then
\begin{equation*}
\frac{|A|}{2p-1}\sum_{i=1}^m \lambda_i \cdot d(P_i)
  \leq \frac{p}{p-1} d(A,A) + d(A,B)\enspace.
\end{equation*}
\end{lemma}
\begin{proof}
For $i\in [m]$ we have
\begin{align*}
d(A,P_i \cap A) + d(A, P_i \cap B) &= d(A, P_i)
  \geq \frac{|P_i|}{|A|} d(A) + \frac{|A|}{|P_i|} d(P_i),
\end{align*}
where the inequality follows from~\eqref{eq:negTypeIneqSet}.
Multiplying both the left-hand side and right-hand side by $\lambda_i$, and summing the resulting inequality over all $i\in [m]$, we get the following inequality, which follows from property~\ref{item:exchangeLoad} of Lemma~\ref{lem:matIntExchangeGen}:
\begin{align*}
\frac{p}{p-1} d(A,A\setminus B) + d(A, B\setminus A)
  &\geq \frac{\frac{p}{p-1} |A\setminus B| + |B\setminus A|}{|A|} d(A)
   + \frac{|A|}{2p-1} \sum_{i=1}^m \lambda_i \cdot d(P_i) \\
  &\geq 
   \frac{|A|}{2p-1} \sum_{i=1}^m \lambda_i \cdot d(P_i).
\end{align*}
The desired statement now follows because $d(A,A) \geq d(A,A\setminus B)$ and $d(A, B) \geq d(A,B\setminus A)$.
\end{proof}

The following lemma shows that the local search algorithm will go towards a set $A$ such that $\sqrt{d(A)}$ approaches
$(1-\frac{2}{p-1}-\frac{24p}{k})\sqrt{d(\OPT)}$, and hence $d(A)$ approaches $(1-\frac{2}{p-1}-\frac{24p}{k})^2 d(\OPT)$.
Notice that for large enough $k$, one can choose a $p$ such that this approximation factor gets arbitrarily small.

\begin{lemma}\label{lem:matIntAverageImp}
Let $p\in \mathbb{Z}_{\geq 2}$ with $8p \leq k$, where $k$ is the cardinality of a maximum cardinality set in $\mathcal{I}_1\cap \mathcal{I}_2$. Let $A\in \mathcal{I}_1\cap \mathcal{I}_2$ with $|A|\geq \sfrac{k}{2}$, and let $P_i\subseteq X$ for $i\in [m]$ and $\lambda_i > 0$ for $i\in [m]$ be a family of sets and coefficients satisfying the conditions of Lemma~\ref{lem:matIntExchangeGen} for the sets $A$ and $\OPT$. Then 
\begin{align*}
\frac{1}{\sum_{i=1}^m \lambda_i}&\sum_{i=1}^m \lambda_i \cdot
   \left(d(A \Delta P_i) - d(A)\right)\\
&\geq
\begin{cases}
  \frac{1}{3k}\sqrt{d(A)} \left( \sqrt{d(\OPT)} \cdot (1-\frac{2}{p-1}-\frac{24p}{k}) - \sqrt{d(A)}\right) &\text{if } d(A) > \frac{1}{50}d(\OPT),\\
   \frac{1}{24k} d(\OPT) &\text{if } d(A) \leq \frac{1}{50}d(\OPT)\enspace .
\end{cases}
\end{align*}
\end{lemma}

\begin{proof}
For brevity, let $S_i=P_i\cap A$ and $T_i = P_i \cap \OPT$ for $i\in [m]$. We have for $i\in [m]$,
\begin{align*}
d(A\Delta P_i) &= d((A\setminus S_i) \cup T_i) = d(A\setminus S_i) + d(T_i) + d(A\setminus S_i, T_i)\\
  &= d(A) + d(S_i) - d(A,S_i) + d(T_i) + d(A,T_i) - d(S_i,T_i)\\
  &\geq d(A) - d(A,S_i) + d(A,T_i) - d(S_i, T_i)\enspace.
\end{align*}
Let $\alpha = \frac{2p-1}{|A|}$ and $\lambda= \sum_{i=1}^m \lambda_i$.

Multiplying the above inequality by $\lambda_i$ and summing it over all $i\in [m]$, we obtain the inequality below, which follows from property~\ref{item:exchangeLoad} of Lemma~\ref{lem:matIntExchangeGen}:
\begin{equation}\label{eq:preBoundSymDiff}
\begin{aligned}
\sum_{i=1}^m \lambda_i \cdot d(A\Delta P_i)
  &\geq \lambda\cdot d(A) - \frac{p}{p-1} d(A,A\setminus \OPT)
     + d(A, \OPT\setminus A) - \sum_{i=1}^m \lambda_i \cdot d(S_i, T_i)\\
  &\geq \lambda\cdot d(A) - \frac{p}{p-1} d(A,A) + d(A,\OPT) - \sum_{i=1}^m \lambda_i \cdot d(S_i,T_i)\\
  &\geq \lambda\cdot d(A) - \frac{p}{p-1} d(A,A) + d(A,\OPT) - \sum_{i=1}^m \lambda_i \cdot d(P_i)\\
  &\geq \lambda\cdot d(A) - \frac{p}{p-1} (1+\alpha) d(A,A)
                       + (1-\alpha) d(A,\OPT)\\
  &\geq \lambda\cdot d(A) - \frac{2 p}{p-1} (1+\alpha) d(A)
                       + (1-\alpha) \left(\frac{|\OPT|}{|A|} d(A) + \frac{|A|}{|\OPT|} d(\OPT) \right),
\end{aligned}
\end{equation}
where the second-to-last inequality follows by Lemma~\ref{lem:boundWithinPi}, and the last one by~\eqref{eq:negTypeIneqSet}.

We continue the expansion of inequality~\eqref{eq:preBoundSymDiff} in two different ways, depending on whether we are in the case $d(A) \leq \frac{1}{50}d(\OPT)$ or $d(A) > \frac{1}{50}d(\OPT)$.

\medskip

\noindent\textbf{Case $d(A) \leq \frac{1}{50} d(\OPT)$:}

To lower bound $(\sfrac{|\OPT|}{|A|}) d(A) + (\sfrac{|A|}{|\OPT|}) d(\OPT)$, consider the coefficients in front of $d(A)$ and $d(\OPT)$. For brevity, let $q=\sfrac{|\OPT|}{|A|}$.
Notice that $q + q^{-1} \geq 2$, no matter how large $|A|$ and $|\OPT|$ are, assuming they are both at least one.
Moreover, since $\OPT$ and $A$ are maximal sets in $\mathcal{I}_1\cap \mathcal{I}_2$, we have $\frac{k}{2} \leq |A|,|\OPT| \leq k$, which implies $q,q^{-1} \geq \frac{1}{2}$. Now consider the expression $q_1 \cdot d(A) + q_2 \cdot d(\OPT)$ under the constraints $q_1+q_2 \geq 2$ and $q_2\geq \sfrac{1}{2}$. This expression is minimized for $q_1=\sfrac{3}{2}$ and $q_2=\sfrac{1}{2}$ because $d(\OPT) \geq d(A)$, and hence
\begin{equation*}
\frac{|\OPT|}{|A|} d(A) + \frac{|A|}{|\OPT|} d(\OPT) \geq \frac{3}{2} d(A) + \frac{1}{2} d(\OPT)\enspace.
\end{equation*}

Using the above inequality, we can now further expand~\eqref{eq:preBoundSymDiff} as shown below. Moreover, we use $\alpha = \frac{2p-1}{|A|} \leq \frac{1}{2}$, which follows from $|A|\geq \frac{k}{2}$ and $8p\leq k$.
\begin{align*}
\sum_{i=1}^{m} \lambda_i \left(d(A\Delta P_i) - d(A) \right)
   &\geq -\frac{2p}{p-1}(1+\alpha) d(A) + (1-\alpha) \left(\frac{3}{2}d(A) + \frac{1}{2} d(\OPT) \right)\\
   &\geq -\frac{3p}{p-1} d(A) + \frac{3}{4} d(A) + \frac{1}{4} d(\OPT)
      && \text{($\alpha \geq \sfrac{1}{2}$)}\\
   &\geq -6 d(A) + \frac{3}{4} d(A) + \frac{1}{4}d(\OPT)
      && \text{($p\geq 2$)}\\
   &= -\frac{21}{4} d(A) + \frac{1}{4} d(\OPT)\\
   &\geq \frac{1}{8} d(\OPT)\enspace.
      &&\hspace*{-2cm}\text{(using $d(A) \leq \sfrac{1}{50}\cdot d(\OPT)$)}
\end{align*}

The result for this case now follows by dividing both sides by $\lambda$ and observing that $\lambda \leq 3k$, which holds due to
\begin{align*}
\lambda &\leq \left\|\sum_{i=1}^m \lambda_i \chi^{P_i} \right\|_1
    && \text{(using $P_i\neq\emptyset$ for $i\in [m]$)}\\
    &= \left\|\frac{p}{p-1}\chi^{A\setminus B} + \chi^{B\setminus A} \right\|_1
        && \text{(by point~\ref{item:exchangeLoad} of Lemma~\ref{lem:matIntExchangeGen})}\\
   &\leq \frac{p}{p-1} k + k  \leq 3k.
\end{align*}
\medskip

\noindent\textbf{Case $d(A) > \frac{1}{50} d(\OPT)$:}

The minimum of the function $f(q) = q\cdot d(A) + \sfrac{1}{q}\cdot d(\OPT)$ for $q >0$ is achieved at $q=\sqrt{\sfrac{d(B)}{d(A)}}$, and leads to a value of $2 \sqrt{d(A) d(\OPT)}$. Hence
\begin{equation*}
\frac{|\OPT|}{|A|} d(A) + \frac{|A|}{|\OPT|} d(\OPT) \leq 2 \sqrt{d(A) d(\OPT)}.
\end{equation*}
Using the above inequality to further expand~\eqref{eq:preBoundSymDiff} we obtain
\begin{equation}\label{eq:preBoundSymDiff2}
\begin{aligned}
\sum_{i=1}^m \lambda_i (d(A\Delta P_i) - d(A))
   &\geq - \frac{2p}{p-1} (1+\alpha) d(A)
      + 2 (1-\alpha) \sqrt{d(A) d(\OPT)}\\
  &= \sqrt{d(A)} \left( \left(1-\frac{2p}{p-1}(1+\alpha)\right)\sqrt{d(A)} + 2(1-\alpha) \sqrt{d(\OPT)} - \sqrt{d(A)} \right)\\
  &\geq \sqrt{d(A)} \left( \left(1-\frac{2p}{p-1}(1+\alpha) + 2(1-\alpha) \right)\sqrt{d(\OPT)} - \sqrt{d(A)} \right),
\end{aligned}
\end{equation}
where we used $d(A) \leq d(\OPT)$ and $1-\frac{2p}{p-1}(1+\alpha) < 0$ for the last inequality.
It remains to simplify the coefficient of $\sqrt{d(\OPT)}$ within the parentheses.
\begin{align*}
-\frac{2p}{p-1}(1+\alpha) + 2(1-\alpha) &= -\frac{2}{p-1}(1+\alpha) - 4\alpha\\
  &= -\frac{2}{p-1} - \frac{2p-1}{|A|} \cdot \left(\frac{2}{p-1} + 4\right)
        && \text{(using $\alpha = \sfrac{(2p-1)}{|A|}$)} \\
  &\geq -\frac{2}{p-1} -\frac{2p}{|A|} \cdot \frac{4p-2}{p-1}\\
  &\geq -\frac{2}{p-1} - \frac{12p}{|A|}
        && \text{($\sfrac{(4p-2)}{(p-1)}\leq 6$ for $p\geq 2$)}\\
  &\geq -\frac{2}{p-1} - \frac{24p}{k}\enspace.
        && \text{(using $|A|\geq \sfrac{k}{2}$)}
\end{align*}
The result for this case now follows by using the above inequality in~\eqref{eq:preBoundSymDiff2}, dividing both sides by $\lambda$, which satisfies $\lambda\leq 3k$ as shown in the first case.
\end{proof}

The next theorem shows that convergence happens fast. This readily follows by Lemma~\ref{lem:matIntAverageImp}, because as long as $d(A)$ is small, the second case of Lemma~\ref{lem:matIntAverageImp} guarantees an improvement of at least $(\sfrac{1}{24k})d(\OPT)$. Hence, this case can only happen during $O(k)$ iterations. Later, the first case of Lemma~\ref{lem:matIntAverageImp} applies, which guarantees a very fast convergence towards $(1-\frac{2}{p-1}-\frac{24p}{k})^2 d(\OPT)$.

\begin{theorem}\label{thm:convLocalSearchMatInt}
Let $p\in \mathbb{Z}_{\geq 2}$ such that $8p\leq k$, where $k$ is the cardinality of a maximum cardinality set in $\mathcal{I}_1 \cap \mathcal{I}_2$, and $1 - \frac{2}{p-1} - \frac{24p}{k} > 0$.
Let $A\in \mathcal{I}_1\cap \mathcal{I}_2$ with $|A|\geq \sfrac{k}{2}$. Then, letting Algorithm~\ref{alg:localSearchMatInt} run for $\ell \geq k$ iterations with starting set $A$ and parameter $p$, leads to a set $A_\ell \in \mathcal{I}_1\cap \mathcal{I}_2$ satisfying
\begin{equation*}
d(A_{\ell}) \geq 
\left(
   \left(1-\frac{2}{p-1} - \frac{24p}{k}\right)^2
- 2 \left(1 - \frac{1}{60 k} \right)^{\ell-k} \right)
d(\OPT)\enspace.
\end{equation*}
\end{theorem}

\begin{proof}
For brevity, let $\beta = 1-\frac{2}{p-1} - \frac{24p}{k}$.
Let $A_0 = A$ be the set at the beginning of the algorithm, and for $j\in [\ell]$ we denote by $A_j\in \mathcal{I}_1\cap \mathcal{I}_2$ the common independent set after $j$ iterations of Algorithm~\ref{alg:localSearchMatInt}. 

Lemma~\ref{lem:matIntAverageImp} shows a lower bound on the average improvement induced by the exchanges $A\Delta P_i$. Since $A_{j+1}$ corresponds to the best exchange for parameter $p$, the difference $d(A_{j+1}) - d(A_j)$ is at least as high as this lower bound.
Thus, as long as $d(A_j)\leq \frac{1}{50}d(\OPT)$, the second case of Lemma~\ref{lem:matIntAverageImp} implies that we have $d(A_{j+1})) \geq d(A_j) + \frac{1}{24k} d(\OPT)$.
Hence after at most $k$ steps, we have a set of value strictly more than~$\frac{1}{50} d(\OPT)$, i.e.,
\begin{equation}\label{eq:lateAjsLarge}
d(A_j) > \frac{1}{50} d(\OPT) \qquad \forall j\in \{k,\ldots \ell\}\enspace.
\end{equation}
hus, for the iterations $k+1, k+2 ,\ldots, \ell$, the first case of Lemma~\ref{lem:matIntAverageImp} applies, which implies for $j\in \{k,\ldots, \ell-1\}$:
\begin{align*}
d(A_{j+1}) - d(A_j) &\geq \frac{1}{3k} \sqrt{d(A_j)}
  \left( \beta \sqrt{d(\OPT)} - \sqrt{d(A_j)} \right)\\
  &\geq \frac{1}{30 k} \sqrt{d(\OPT)}
  \left( \beta \sqrt{d(\OPT)} - \sqrt{d(A_j)} \right),
\end{align*}
where the second inequality follows from~\eqref{eq:lateAjsLarge}.
Dividing both sides by $\sqrt{d(A_{j+1})} + \sqrt{d(A_j)}$ we get
\begin{align*}
\sqrt{d(A_{j+1})} - \sqrt{d(A_j)} &\geq
  \frac{1}{30k} \frac{\sqrt{d(\OPT)}}{\sqrt{d(A_{j+1})} + \sqrt{d(A_j)}}
  \left( \beta \sqrt{d(\OPT)} - \sqrt{d(A_j)} \right)\\
 &\geq \frac{1}{60k} 
  \left( \beta \sqrt{d(\OPT)} - \sqrt{d(A_j)} \right)\enspace,
\end{align*}
where the second inequality follows from $d(A_{j+1}), d(A_j)\leq d(\OPT)$.
The above inequality can be rewritten as follows
\begin{align*}
\beta \sqrt{d(\OPT)} - \sqrt{d(A_{j+1})}
   &\leq \left( 1- \frac{1}{60k} \right)
          \left(\beta \sqrt{d(\OPT)} - \sqrt{d(A_j)} \right)
   \qquad \forall j\in \{k,\ldots, \ell-1\}\enspace.
\end{align*}
Hence,
\begin{align*}
\beta \sqrt{d(\OPT)} - \sqrt{d(A_{\ell})}
   &\leq \left( 1- \frac{1}{60k} \right)^{\ell-k}
          \left(\beta \sqrt{d(\OPT)} - \sqrt{d(A_k)} \right)\\
   &\leq \left( 1- \frac{1}{60k} \right)^{\ell-k}
          \beta \sqrt{d(\OPT)}\enspace. 
\end{align*}
Multiplying both sides with $\beta\sqrt{d(\OPT)} + \sqrt{d(A_\ell)}$ we get
\begin{align*}
\beta^2 d(\OPT) - d(A_{\ell})
   &\leq \left( 1- \frac{1}{60k} \right)^{\ell-k}
          \beta \sqrt{d(\OPT)} \left(\beta \sqrt{d(\OPT)} + \sqrt{d(A_\ell)}\right)\\
   &\leq \left( 1- \frac{1}{60k} \right)^{\ell-k} d(\OPT) \beta(\beta+1)\\
   &\leq 2 \left( 1- \frac{1}{60k} \right)^{\ell-k} d(\OPT)\enspace ,
      && \text{(using $\beta\leq 1$)}
\end{align*}
which implies the desired result.
\end{proof}

To obtain a PTAS, and therefore prove Theorem~\ref{thm:PTASmatInt}, it remains to choose the right value for $p$ in Theorem~\ref{thm:convLocalSearchMatInt} if $k$ is large enough. For small $k$, one can simply enumerate over all (polynomially many) feasible solutions.
This leads to Theorem~\ref{thm:PTASmatInt}, which we repeat below for completeness.

\begin{reptheorem}{thm:PTASmatInt}
There is a PTAS for \MSD with negative type distances subject to a matroid intersection constraint.
\end{reptheorem}

\begin{proof}
Let $M_j=(X,\mathcal{I}_j)$ for $j\in \{1,2\}$ be the two matroids imposing the matroid intersection constraint to the \MSD problem.
As usual, let $k$ be the cardinality of a maximum cardinality set in $\mathcal{I}_1\cap \mathcal{I}_2$.
Let $\epsilon > 0$ be our error tolerance, i.e., we want to obtain a $(1-\epsilon)$-approximation.

If $k < \frac{144}{\epsilon}
  \left( \lceil \frac{12}{\epsilon} \rceil +1 \right)$, then we can enumerate over all sets in $\mathcal{I}_1\cap \mathcal{I}_2$, since there are only polynomially many such sets.
In what follows, we therefore assume $k \geq \frac{144}{\epsilon} \left( \lceil \frac{12}{\epsilon} \rceil +1 \right)$. We set $p=\lceil\frac{12}{\epsilon}\rceil + 1$ and let $A$ be any set in $\mathcal{I}_1\cap \mathcal{I}_2$ of cardinality $|A|\geq \sfrac{k}{2}$. Now we let Algorithm~\ref{alg:localSearchMatInt} run with starting set $A$, parameter $p$, and for a number $\ell \geq k$ of iterations satisfying
\begin{equation}\label{eq:downToE3}
2 \left( 1- \frac{1}{60 k} \right)^{\ell - k} \leq \frac{\epsilon}{3}\enspace.
\end{equation}
Notice that the above inequality is satisfied for a value $\ell=O(k \log \sfrac{1}{\epsilon})$, which is polynomially bounded. Since moreover $p = O(1)$, Algorithm~\ref{alg:localSearchMatInt} runs in polynomial time returning some set $A_\ell \in \mathcal{I}_1\cap \mathcal{I}_2$.
By Theorem~\ref{thm:convLocalSearchMatInt}, the dispersion of $A_\ell$ satisfies
\begin{align*}
d(A_{\ell}) &\geq \left(
\left( 1- \frac{2}{p-1} - \frac{24p}{k} \right)^2
- 2 \left( 1- \frac{1}{60k} \right)^{\ell - k} \right)
d(\OPT)\\
  &\geq \left(
      \left( 1- \frac{2}{p-1} - \frac{24p}{k} \right)^2
          -\frac{\epsilon}{3} \right)
      d(\OPT)
     && \text{(by~\eqref{eq:downToE3})}\\
  &\geq \left(
      \left( 1- \frac{\epsilon}{6} - \frac{24(\lceil\frac{12}{\epsilon}\rceil+1)}{k} \right)^2
 - \frac{\epsilon}{3}
\right)
 d(\OPT)
     && \text{(using~$p=\lceil\sfrac{12}{\epsilon}\rceil + 1$)}\\
  &\geq \left(
       \left(1-\frac{\epsilon}{3}\right)^2 - \frac{\epsilon}{3}
        \right) d(\OPT)
     && \text{(using $k\geq \sfrac{144}{\epsilon} (\lceil\sfrac{12}{\epsilon}\rceil + 1)$)}\\
  &\geq \left( 1 - \epsilon \right) d(\OPT)\enspace,
\end{align*}
as desired.
\end{proof}

\section{Combinations with submodular objectives} 
\label{sec:submodObj}

In this section, we are interested in generalized diversity measures, which extend the dispersion function considered previously, by adding a monotone submodular function $f: 2^X \rightarrow \mathbb{Q}_{\geq 0}$ to it. Our goal is to maximize the objective $g(A) = d(A) + f(A)$ subject to a matroid constraint $M=(X,\mathcal{I})$. For brevity, we denote this problem by \MSDf. Throughout this section, we assume that $f$ is a nonnegative monotone submodular function on $X$ that is normalized, i.e., $f(\emptyset) = 0$. Notice that $f$ being normalized is an assumption without loss of generality, because any (non-normalized) function $f$ can be replaced by the normalized submodular function $h(S) = f(S) - f(\emptyset)$.

Submodular functions capture many natural diversity functions. One classic example is to count how many different ``types'' are covered by a set of elements. Here, one can define an arbitrary family $T_1,\dots,T_p \subseteq X$ of subsets of the ground set, where each $T_i$ can be thought of items sharing a particular property, and elements in $T_i$ are said to be of type $T_i$. Notice that the same element can be part of many different types. The function that assigns to each set $A\subseteq X$, the number of different types that are contained in $A$ is a well-known example of a submodular function, which is called a coverage function.

The problem of maximizing the sum of the dispersion and a monotone submodular function has previously been considered by Borodin et al.~\cite{borodin2012max}, who considered metric distance spaces. We recall that for metric distance spaces, a locally optimal solution is (only) a $\sfrac{1}{2}$-approximation. As we briefly explain below, our stronger analysis for negative type distances is compatible with their analysis and leads to stronger results for the case of negative type distances. Moreover, we show that these results can be combined with recent results on local search algorithms for maximizing submodular functions, leading to approximation guarantees that are asymptotically optimal for a wide range of problems.

For simplicity, we will focus on the approximation guarantee of solutions that are locally optimal with respect to our local search algorithm, without going into details on how to efficiently compute a solution whose objective value is no more than a factor $1+\epsilon$ larger than the objective value of a locally optimal solution, for a fixed error $\epsilon >0$. The steps to show fast convergence of the local search algorithm are standard, and can also be done analogously to our analysis in Section~\ref{sec:matConstraints}.

We start by giving a very brief overview of some recent results on which we build up. More precisely, we first give a quick overview of local search results by Borodin et al.~\cite{borodin2012max}, which we strengthen in the following. We then recap a non-oblivious local search algorithm by Filmus and Ward~\cite{filmus2014monotone} to obtain an optimal approximation guarantee for submodular maximization under a matroid constraint, and an extension of it by Sviridenko et al.~\cite{sviridenko2015optimal} to submodular functions with bounded curvature.
Our strengthened results then readily follow by combining these previous approaches together with our stronger analysis for distances of negative type.

Throughout this section, $A,B\subseteq X$ are bases of $M$, and $\pi: A \rightarrow B$ is a Brualdi bijection between $A$ and $B$, if not stated otherwise.

\subsection*{Local search results by Borodin et al.~\cite{borodin2012max}}

Borodin et al.~\cite{borodin2012max} study the \MSDf problem for metric distance spaces. They show that the local search algorithm achieves an approximation guarantee of $\frac{1}{2}$ for each objective function individually, and they prove that the same result carries over to the combined objective function. More specifically, in \cite[Lemmas 5 and 7]{borodin2012max} it is proven that,\footnote{Some results cited in this section are originally stated in less generality, but their original proofs carry on directly to the more general versions.} if $d$ is metric, and $f$ is submodular, monotone and normalized, then

\begin{align}
d(A)&\geq \frac{1}{2}d(B)+\frac{1}{2}\sum_{a\in A}\left(d(A)-d(A-a+\pi(a))\right), \ \text{ and} \label{ineq:metric} \\
f(A)&\geq \frac{1}{2}f(B)+\frac{1}{2}\sum_{a\in A}\left(f(A)-f(A-a+\pi(a))\right). \label{ineq:submod}
\end{align}
Inequality~\eqref{ineq:metric} only holds if the rank $k$ of the underlying matroid is at least $3$. For simplicity, we assume $k\geq 3$ in what follows.
Summing up~\eqref{ineq:metric} and~\eqref{ineq:submod}, one immediately obtains 
$$g(A)\geq \frac{1}{2}g(B)+\frac{1}{2}\sum_{a\in A}\left(g(A)-g(A-a+\pi(a))\right).$$

From each of these lines, one can easily derive that any local optimum has an approximation guarantee of $\frac{1}{2}$ for the corresponding objective function. This approximation is tight for function $d$, as we mentioned in the introduction, assuming that the planted clique problem is hard; and for function $f$ this approximation is known to match the locality gap of this local search algorithm.

\subsection*{Non-oblivious local search by Filmus and Ward~\cite{filmus2014monotone}}

When trying to maximize a function $f$, sometimes it is convenient to define an auxiliary potential function $F:2^X \rightarrow \mathbb{R}_{\geq 0}$, and run a local search algorithm over $F$ instead of $f$. More precisely, an exchange step is done if it leads to a strict improvement of the function value of $F$ instead of $f$. Such an algorithm is called \emph{non-oblivious}.
Filmus and Ward \cite[Theorem 5.1]{filmus2014monotone} prove that there is a potential function $F: 2^X \rightarrow \mathbb{R}_{\geq 0}$ for $f$ such that
\begin{equation}\label{ineq:filmus}
f(A)\geq \left(1-\frac{1}{e}\right)f(B)+\left(1-\frac{1}{e}\right)\sum_{a\in A}\left(F(A) - F(A-a+\pi(a))\right)\enspace.
\end{equation}
From this, they conclude that the non-oblivious algorithm associated to $F$ offers an approximation guarantee of $1-\frac{1}{e}-\epsilon$ for the matroid-constrained maximization of a monotone, normalized submodular function. This is best possible, as it is proven by Feige \cite{feige1998threshold} that improving the bound $1-\frac{1}{e}$ is NP-hard even if $f$ is an explicitly given coverage function. And Nemhauser and Wolsey \cite{nemhauser1978best} show that improving upon this bound requires an exponential number of queries in the value oracle model.

We remark that an exact evaluation of the potential function $F$ defined in~\cite{filmus2014monotone} requires a superpolynomial number of evaluations of $f$. However, the authors prove that all the necessary evaluations of $F$ during the execution of the algorithm can be approximated efficiently, in such a way that with high probability, the ensuing loss in the approximation ratio is arbitrarily small. For clarity of exposition, we assume in this paper that we have access to a value oracle for $F$. It is also proven in~\cite{filmus2014monotone} that $F$ has further properties which are sufficient for the non-oblivious algorithm to converge fast.

\subsection*{Submodular maximization with bounded curvature by Sviridenko et al.~\cite{sviridenko2015optimal}}

Conforti and Cornu\'ejols \cite{conforti1984submodular} define the \emph{curvature} of a monotone submodular function $f$ as 
$$c=1-\min_{x\in X}\frac{f(X)-f(X-x)}{f(x)-f(\emptyset)}.$$
The curvature is a coefficient between 0 and 1 that measures how close the function is to being linear, where $c=0$ corresponds to a linear function and $c=1$ to an arbitrary submodular function. For a monotone submodular function $f$ of curvature $c$, Sviridenko et al.~\cite{sviridenko2015optimal} consider the decomposition $f=l+f'$, where $l(A)=f(\emptyset)+\sum_{a\in A} (f(X)-f(X-a))$, and $f'(A)=f(A)-l(A)$, for each $A\subseteq X$. They prove that $l$ is linear, $f'$ is submodular, monotone and normalized, and $f'(A)\leq c \cdot f(A)$ for each $A\subseteq X$. It is easy to see that for any linear function $l$,
\begin{equation}\label{eq:linear}
l(A)= l(B)+\sum_{a\in A}\left(l(A)-l(A-a+\pi(a))\right).
\end{equation}
 
Hence, if we define a potential function $F'$ for $f'$ as in inequality (\ref{ineq:filmus}) (see \cite{filmus2014monotone}), and use it to define the potential function $F=l+\left(1-\frac{1}{e}\right)F'$ for $f$, the sum of inequality (\ref{ineq:filmus}) and equation (\ref{eq:linear}) gives for any $A,B\in \mathcal{I}$
 \begin{align}\label{ineq:sviridenko}
f(A) &\geq f(B) -\frac{1}{e}f'(B)+\sum_{a\in A}\left(F(A) - F(A-a+\pi(a))\right)\nonumber\\
 & \geq \left(1-\frac{c}{e}\right)f(B)+\sum_{a\in A}\left(F(A)-F(A-a+\pi(a))\right).
\end{align}
Thus, in \cite{sviridenko2015optimal} they conclude that, for a monotone submodular function $f$ of curvature $c$,  the local search algorithm associated to the potential function $F$ offers an approximation guarantee of $1-\frac{c}{e} - \epsilon$. Moreover, they extend the negative result of \cite{nemhauser1978best} to prove that this bound is best possible; namely they prove that, for each $c\in [0,1]$, improving upon the bound of $1-\frac{c}{e}$ requires an exponential number of queries in the value oracle model.

\subsection*{Putting things together}

By combining our stronger analysis for local search subject to negative distance spaces with the results highlighted above, we obtain our main result which leads to Theorem~\ref{thm:msd+f}.

\begin{theorem}\label{thm:d+f_NT}
Consider \MSDf over a rank $k$ matroid constraint, where $d$ is of negative type and $f$ has curvature $c$. Let $\lambda_d=\frac{d(\OPT)}{g(\OPT)}$ and $\lambda_f=\frac{f(\OPT)}{g(\OPT)}$. Then, there is a non-oblivious local search algorithm such that any locally optimal solution has an approximation guarantee of 
$$1-\lambda_d \frac{4}{k}-\lambda_f \frac{c}{e}\geq 1-\max\left\{\frac{4}{k}, \frac{c}{e}\right\}.$$ 
\end{theorem}

\begin{proof}
For the function $f$, consider the potential function $F$ as defined in inequality (\ref{ineq:sviridenko}) (see \cite{sviridenko2015optimal}); and for $g = d + f$, define the potential function $G=\left(1-\frac{2}{k}\right)d + F$.
Notice that inequality~\eqref{eq:matAverageImpDetailed} implies
\begin{equation*}
d(A) \geq  \left(1-\frac{4}{k}\right) d(\OPT) +
\left(1-\frac{2}{k}\right)\sum_{a\in A}\left( d(A) - d(A-a+\pi(a))\right)\enspace,
\end{equation*}
which follows by multiplying~\eqref{eq:matAverageImpDetailed} by $1-\frac{2}{k}$, rearranging terms, and using $(1-\sfrac{2}{k})^2 \geq 1-\sfrac{4}{k}$, and $(1+\sfrac{2}{k})\cdot (1-\sfrac{2}{k}) \leq 1$.

Combining the above inequality with \eqref{ineq:sviridenko}, we obtain for any basis $A$ of $M$
\begin{equation*}
g(A)\geq g(\OPT) - \frac{4}{k} d(\OPT)-\frac{c}{e} f(\OPT) +\sum_{a\in A} \left(G(A)-G(A-a+\pi(a))\right)\enspace,
\end{equation*}
where $\pi:A\rightarrow \OPT$ is a Brualdi bijection between $A$ and $\OPT$.
We run the non-oblivious local search algorithm with respect to the potential $G$. Let $A$ be a local optimum with respect to $G$. The local optimality of $A$ implies $G(A)-G(A-a+\pi(a))\geq 0$, which, by the above inequality, implies
\begin{equation*}
g(A)\geq \left(1-\lambda_d \frac{4}{k}-\lambda_f \frac{c}{e}\right)g(\OPT)\enspace.
\end{equation*}
To complete the proof, notice that this approximation ratio is a convex combination of $1-\frac{4}{k}$ and $1-\frac{c}{e}$, and hence it is larger than the smaller of the two values.
\end{proof}

For completeness, we highlight that even for metric distances $d$, the recent non-oblivious local search procedures presented in \cite{filmus2014monotone,sviridenko2015optimal} lead to a strengthening of the result of Borodin et al.~\cite{borodin2012max}.

\begin{theorem}\label{thm:d+f_metric}
Consider \MSDf over a matroid constraint, where $d$ is metric, and $f$ has curvature $c$. Let $\lambda_d\geq 0$ such that $\frac{d(O)}{g(O)}\leq \lambda_d$. Then, there is a non-oblivious local search algorithm such that any locally optimal solution has an approximation guarantee of
\begin{equation*}
1-\lambda_d \frac{1}{2}-(1-\lambda_d)\frac{c}{e}\enspace.
\end{equation*}
\end{theorem}
\begin{proof}
The proof is analogous to the previous one, except that we add up inequalities (\ref{ineq:metric}) and (\ref{ineq:sviridenko}) instead of the inequality provided by Lemma~\ref{lem:matAverageImp} and~\eqref{ineq:sviridenko}; and consequently for $g=d+f$ we define the potential function $G=\frac{1}{2}d + F.$
\end{proof}

\bibliographystyle{plain}
\bibliography{lit}

\end{document}

%% file: dAmapb.tikz
\begin{tikzpicture}[scale=0.8]

\def\dist{5cm}
\def\hdist{3.5cm}

\tikzstyle{ns}=[draw,circle,inner sep=1]
\tikzstyle{es}=[blue, very thick]

\newcommand\basegraph{
\begin{scope}[thick]

\begin{scope}[every node/.style=ns]
\node (v1) at (0,0) {};
\node (v2) at (0.7,0.3) {};
\node (v3) at (-0.1,0.6) {};
\node (v4) at (0.3,0.9) {};

\node (a) at (1.8,-0.5) {};
\node (b) at (1.8,1.3) {};
\end{scope}

\draw (0.2,0.4) circle (1);

\node at (0.2,1.4) [above] {$A-a$};
\node at (a) [right] {$a$};
\node at (b) [right] {$b$};

\end{scope}
}


\begin{scope}
\basegraph
\begin{scope}[es]
\draw (v1) -- (v2);
\draw (v1) -- (v3);
\draw (v1) -- (v4);
\draw (v2) -- (v3);
\draw (v2) -- (v4);
\draw (v3) -- (v4);

\foreach \i in {1,...,4} {
\draw (b) -- (v\i);
}
\end{scope}

\node[blue] at (0.2,-1) {$d(A-a+b)$};

\end{scope}

\begin{scope}[xshift=\dist]
\basegraph
\begin{scope}[es]
\draw (v1) -- (v2);
\draw (v1) -- (v3);
\draw (v1) -- (v4);
\draw (v2) -- (v3);
\draw (v2) -- (v4);
\draw (v3) -- (v4);

\foreach \i in {1,...,4} {
\draw (a) -- (v\i);
}
\end{scope}

\node[blue] at (0.2,-1) {$d(A)$};

\end{scope}

\begin{scope}[xshift=2*\dist]
\basegraph
\begin{scope}[es]

\foreach \i in {1,...,4} {
\draw (b) -- (v\i);
}
\draw (b) -- (a);
\end{scope}

\node[blue] at (0.2,-1) {$d(A,b)$};

\end{scope}

\begin{scope}[shift={(0.5*\dist,-\hdist)}]
\basegraph
\begin{scope}[es]
\draw (a) -- (b);
\end{scope}

\node[blue] at (0.2,-1) {$d(a,b)$};

\end{scope}

\begin{scope}[shift={(1.5*\dist,-\hdist)}]
\basegraph
\begin{scope}[es]
\foreach \i in {1,...,4} {
\draw (a) -- (v\i);
}
\end{scope}

\node[blue] at (0.2,-1) {$d(a,A)$};

\end{scope}

\end{tikzpicture}